\documentclass[11pt, journal, draftclsnofoot, onecolumn]{IEEEtran}

\usepackage{mathtools,enumitem,color,bbm}
\usepackage{hyperref}
\usepackage{subfig}

\usepackage{amsthm,mathrsfs}
\usepackage[T1]{fontenc}
\usepackage{cite}
\usepackage{txfonts}
\DeclareMathOperator*{\argmax}{arg\,max}
\DeclareMathOperator*{\argmin}{arg\,min}
\IEEEoverridecommandlockouts
\usepackage{graphicx}

\definecolor{green}{rgb}{0.6627,0.8196,0.5568}
\definecolor{blue}{rgb}{0.6875,0.8750,0.8984}
\definecolor{purple}{rgb}{ 0.7647,    0.6078,    0.8824}

\newcommand{\discardpages}[1]{% \discardpages{<csv list>}
  \xdef\discard@pages{#1}% Store pages to discard
  \AtBeginShipout{% At shipout, decide whether to discard page/not
    \renewcommand*{\do}[1]{% How to handle each page entry in csv list
      \ifnum\value{page}=##1\relax%
        \AtBeginShipoutDiscard% Discard page/not
        \gdef\do####1{}% Do nothing further
      \fi%
    }%
    \expandafter\docsvlist\expandafter{\discard@pages}% Process list of pages to discard
  }%
}
\newif\ifkeeppage
\newcommand{\keeppages}[1]{% \keeppages{<csv list>}
  \xdef\keep@pages{#1}% Store pages to keep
  \AtBeginShipout{% At shipout, decide whether to discard page/not
    \keeppagefalse%
    \renewcommand*{\do}[1]{% How to handle each page entry in csv list
      \ifnum\value{page}=##1\relax%
        \keeppagetrue% Page should be kept
        \gdef\do####1{}% Do nothing further
      \fi%
    }%
    \expandafter\docsvlist\expandafter{\keep@pages}% Process list of pages to keep
    \ifkeeppage\else\AtBeginShipoutDiscard\fi% Discard page/not
  }%
}

\newtheorem{Theorem}{Theorem}
\newtheorem{Proposition}{Proposition}

\newtheorem{Remark}{Remark}
\newtheorem{Definition}{Definition}

\ifCLASSINFOpdf

\begin{document}
%
% paper title
% Titles are generally capitalized except for words such as a, an, and, as,
% at, but, by, for, in, nor, of, on, or, the, to and up, which are usually
% not capitalized unless they are the first or last word of the title.
% Linebreaks \\ can be used within to get better formatting as desired.
% Do not put math or special symbols in the title.
% make the title area

%---------------------------------------------

%\begin{titlepage}

\title{Optimal Fault-Tolerant Data Fusion in Sensor Networks: Fundamental Limits and Efficient Algorithms}

% author names and affiliations
% use a multiple column layout for up to three different
% affiliations
\author{

\IEEEauthorblockN{ Marian  Temprana Alonso, Farhad Shirani \thanks{This work was supported in part by NSF grant CCF-2241057.}, Sitharama S. Iyengar,
\\Florida International University,
\\Email: mtemp009@fiu.edu, fshirani@fiu.edu, iyengar@cis.fiu.edu}
}

\maketitle

\begin{abstract}
Distributed estimation in the context of sensor networks is considered, where distributed agents are given a set of sensor measurements, and are tasked with estimating a target variable. 
A subset of sensors are assumed to be faulty. The objective is to minimize i) the mean square estimation error at each node (accuracy objective), and ii) the mean square distance between the estimates at each pair of nodes (consensus objective). It is shown that there is an inherent tradeoff between the former and latter objectives. Assuming a general stochastic model, the sensor fusion algorithm optimizing this tradeoff is characterized through a computable optimization problem, and a Cramer-Rao type lower bound for the achievable accuracy-consensus loss is obtained. Finding the optimal sensor fusion algorithm is computationally complex. To address this, a general class of low-complexity Brooks-Iyengar  Algorithms  are introduced, and their performance, in terms of accuracy and consensus objectives, is compared to that of optimal linear estimators through case study simulations of various scenarios.
\end{abstract}

% no keywords

% For peer review papers, you can put extra information on the cover
% page as needed:
% \ifCLASSOPTIONpeerreview
% \begin{center} \bfseries EDICS Category: 3-BBND \end{center}
% \fi
%
% For peerreview papers, this IEEEtran command inserts a page break and
% creates the second title. It will be ignored for other modes.
\IEEEpeerreviewmaketitle
\vspace{-.06in}
\section{Introduction}
Distributed estimation arises naturally in various application scenarios including sensor networks \cite{akyildiz2010wireless,zhong2004combining,he2020distributed,gubner1993distributed}, robotics \cite{ahmed2016distributed}, navigation, tracking and radar
%bar2004estimation,
networks \cite{ yan2020optimal}, and monitoring and surveillance applications \cite{pasqualetti2012distributed}. The problem has been studied extensively in the literature under a range of assumptions and formulations. In recent years, there has been significant renewed interest in distributed estimation scenarios involving sensor fusion due to  advances in sensor and communication technologies, which has enabled the application of massive non-homogeneous collections of sensors including radar, LiDAR, camera, ultrasonic, GPS, IMU, and V2X in applications including tracking, navigation, and autonomous driving \cite{kim2021eagermot}. 

This paper considers the distributed estimation scenario shown in Figure \ref{fig:1}, where  $m$ distributed agents receive a set of $n$ sensor measurements, in the form of confidence intervals $[L_{i,j},U_{i,j}], i\in [n], j\in [m]$, and are tasked with finding the best estimate of a target variable $X$ with respect to fidelity criteria on both individual accuracy and collective consensus among sensors. In general, the sensor measurements are noisy,   and additionally a subset of sensors may be faulty, where a faulty sensor provides measurements which are independent of the target variable. Furthermore, faulty sensors transmit independent measurements to each of the distributed agents.  Faulty sensors model both sensor malfunction as well as adversarial interference in the sensor measurement and transmission phases \cite{ivanov2016attack,shin2017illusion}. The distributed estimation system has two objectives: i) \textbf{Accuracy objective:} to minimize the mean square error of the estimate of the target variable by each agent, and ii) \textbf{Consensus objective:} to minimize the square distance between pairs of estimates of the target variable by each of the distributed agents. The former objective is a local performance objective focusing on the individual accuracy of each sensor, whereas the latter objective is a global performance objective focusing on the collective agreement among distributed sensors. The consensus objective is of interest in applications such as clock
synchronization, robot convergence and gathering, autonomous driving, blockchain technologies, and distributed voting, among others. The consensus objective has been studied extensively in the context of the Byzantine Agreement Problem  \cite{lamport1982byzantine,indrakumari2020consensus,sangwan2022byzantine,he2012byzantine} as well as the study of consensus-based filters \cite{olfati2004consensus,carli2008distributed}.

 \begin{figure}[t]
\centering 
\includegraphics[width=0.5\textwidth]{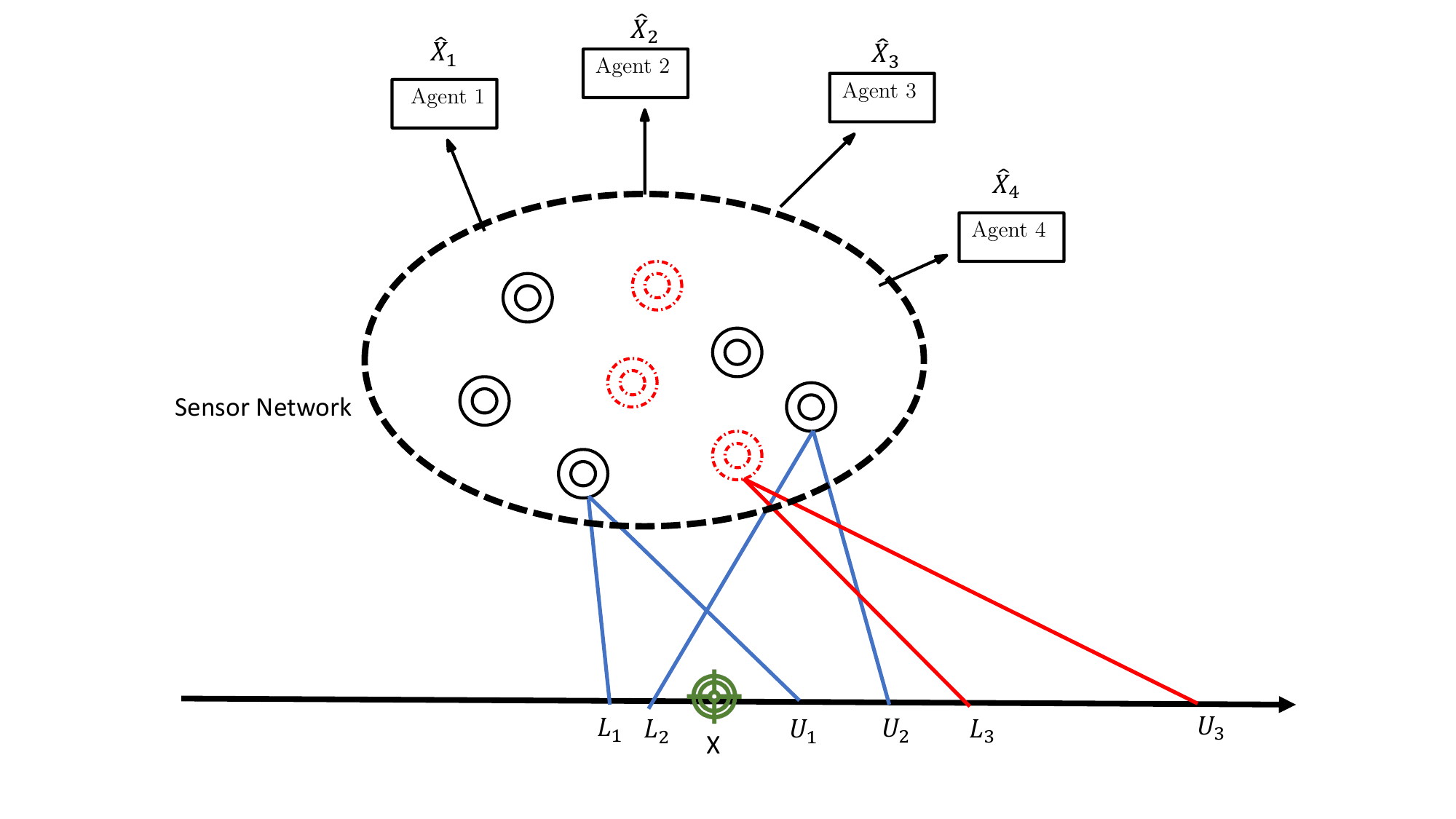}
\caption{Each sensor measures the target variable $X$ and outputs a confidence interval.
Non-faulty sensors are shown by solid black circles and faulty sensors by dashed red circles.  
}
\vspace{-.25in}
\label{fig:1}
\end{figure}

The study of the distributed estimation scenario considered in this paper was initiated in \cite{mahaney1985inexact,dolev1986reaching}. Marzullo proposed a sensor fusion algorithm  in \cite{marzullo1990tolerating} and derived the corresponding worst case performance guarantees. An improved algorithm was proposed by Brooks and Iyengar (BI) in \cite{brooks1996robust}, and worst-case theoretical guarantees for both accuracy and consensus objectives were derived in \cite{ao2016precision}. It was shown that the BI algorithm improves upon Marzullo's algorithm in terms of the worst-case performance in the consensus objective. On the theory side, prior works have considered various other formulations of the distributed estimation problem under Bayesian, mean square error, and Dempster–Shafer theory paradigms  \cite{olfati2004consensus,carli2008distributed,varshney1997multisensor, murphy1998dempster}. In this paper, we first formulate a general distributed estimation problem, and provide an analytical framework to evaluate the performance limits in terms of the accuracy and consensus objectives. We show that there is a fundamental tradeoff between these two objectives, and characterize the fusion operation optimizing the tradeoff in the form of a convex optimization problem.
Furthermore, we provide a lower bound on the achievable accuracy-consensus loss. The bound is quantified by the \textit{average Fisher information} of the sensor output variables and is obtained via a Cramer-Rao type inequality. 
Analytical characterization of the fusion operation optimizing the aforementioned tradeoff requires prior knowledge of the underlying statistics, which is not possible in many real-world applications. Furthermore, even in presence of accurate estimates of the underlying statistics, deriving the optimal decision function has high computational complexity. In the second part of the paper, we propose a generalized class of practical low-complexity Brooks-Iyengar fusion algorithms and evaluate their performance through various simulations of practical scenarios. The main contributions of this work are summarized below:
\vspace{-.03in}
\begin{itemize}[leftmargin=*]
    \item To characterize a tradeoff in the accuracy and consensus objectives and to provide a computable optimization problem for finding the fusion algorithm optimizing this tradeoff.
    \item To provide a Cramer-Rao type bound on the performance of the optimal algorithm in terms of accuracy and consensus.
    \item To introduce a generalized class of BI algorithms (GBI). 
    \item To prove the optimality of the GBI algorithms, in terms of the accuracy objective, under specific statistical assumptions.
    \item To provide case study simulations of the performance of the proposed algorithms.  
\end{itemize}

{\em Notation:}
%The random variable $\mathbbm{1}_{\mathcal{E}}$ is the indicator function of the event $\mathcal{E}$.x
 The set $\{1,2,\cdots, n\}, n\in \mathbb{N}$ is represented by $[n]$. An $n$-length vector is written as $[x_i]_{i\in [n]}$ and an $n\times m$ matrix is written as $[h_{i,j}]_{i,j\in [n]\times [m]}$. We write $\mathbf{x}$ and $\mathbf{h}$ instead of $[x_i]_{i\in [n]}$ and $[h_{i,j}]_{i,j\in [n]\times [m]}$, respectively, when the dimension is clear from context.  The vector $\mathbf{e}_m$ is the $m$-length all-ones vector. Sets are denoted by calligraphic letters such as $\mathcal{X}$. $\Phi$ represent the empty set.  $\mathbb{B}$ denotes the Borel $\sigma$-field. For the event $\mathcal{E}$, the variable $\mathbbm{1}(\mathcal{E})$ denotes the indicator of the event. For random variable $X$ with probability space $(\mathcal{X}, \mathbb{F}_X, P_X)$, the Hilbert space of functions of $X$ with finite variance is denoted by $\mathcal{L}_2(X)$.

\section{Problem Formulation}
\label{sec:form}

We consider a distributed estimation problem involving $m$ agents and $n$ sensors. In the most general formulation, each agent is tasked with estimating the target variable $X$ defined on the probability space $(\mathbb{R}, \mathbb{B}, P_X)$, where $\mathbb{B}$ denotes the Borel $\sigma$-algebra. Each sensor measures the target variable X separately, and the measurement output is assumed to be in the form of an interval $[L,U]$, where $L<U$. The $i$th sensor transmits the pair $(L_{i,j},U_{i,j})$ to the $j$th agent, where $i\in [n], j\in [m]$ and $L_{i,j}\leq U_{i,j}$. This is formalized below. 
\begin{Definition}[\textbf {Distributed Estimation Setup}]
A distributed estimation setup is characterized by the tuple $(n,m, P_{X,({L}_{i,j},{U}_{i,j})_{i\in [n],j\in [m]}})$, where $n,m\in \mathbb{N}$ and $P_{X,({L}_{i,j},{U}_{i,j})_{i\in [n],j\in [m]}}$ is a probability measure defined on $\mathbb{R}^{nm+1}$ such that $P(L_{i,j}\leq U_{i,j}, i\in [n], j\in [m])=0.$
%The variable $X$ is called the target variable, and $({L}_{i,j},{U}_{i,j}), i\in [n], j\in [m]$ represent the transmitted measurement from the $i$th sensor to the $j$th agent.
\end{Definition}

The fusion algorithm is formally defined below.

\begin{Definition}[\textbf{Fusion Algorithm}]
For a distributed estimation setup characterized by the tuple $(n,m,$ $ P_{X,({L}_{i,j},{U}_{i,j})_{i\in [n],j\in [m]}})$, a fusion algorithm $\textbf{f}$ consists of a collection of functions $f_j: \mathbb{R}^{2n}\to \mathbb{R}, j\in [m]$. The estimate of $X$ at the $j$th agent is $\widehat{X}_j\triangleq f_j(L_{i,j},U_{i,j}, i\in [n]), j\in [m]$.
\end{Definition}
The fusion algorithm is evaluated with respect to an accuracy objective and a consensus objective as formalized below.

\begin{Definition}[\textbf{Fusion Objectives}]
For a given a distributed estimation setup $(n,m, P_{X,({L}_{i,j},{U}_{i,j})_{i\in [n],j\in [m]}})$ and fusion algorithm $\textbf{f}= (f_j(\cdot), j\in [m])$, the accuracy is parametrized by the vector $(mse(f_j), j\in [m])$, where:
\begin{align*}
    mse(f_j)\triangleq \mathbb{E}\left(\left(X-\widehat{X}_j\right)^2\right),\quad  j\in [m].
\end{align*}
The consensus is parametrized by the vector $(cns(f_j,f_{j'}), j,j'\in [m])$, where:
\begin{align*}
    cns(f_j,f_{j'})\triangleq \mathbb{E}\left(\left(
    \widehat{X}_j-\widehat{X}_{j'}\right)^2\right),\quad  j,j'\in [m].
\end{align*}
Given parameter $\lambda\in \mathbb{R}$ such that $0\leq \lambda\leq 1$, 
the fusion algorithm $\textbf{f}^*_\lambda$ is called $\lambda$-optimal if it is the solution to the following optimization problem:
\begin{align}
    \textbf{f}^*_\lambda=\argmin_{\textbf{f}=(f_j)_{j\in [m]}}\lambda\!\!\sum_{j\in [m]} mse(f_j)+ \frac{\overline{\lambda}}{m-1}\!\!\sum_{\substack{j,j'\in [m]\\ j\neq j'}}  cns(f_j,f_{j'}), 
    \label{eq:ob}
\end{align}
where the minimum is taken over all $\textbf{f}=(f_j)_{j\in[m]}$ such that $f_j\in \mathcal{L}_2(\prod_{i\in [n]}L_{i,j}\times U_{i,j}), j\in [m]$ and $\overline{\lambda}\triangleq 1-\lambda$.
\end{Definition}

In Section \ref{sec:analytical}, we evaluate the accuracy-consensus tradeoff under the general distributed estimation model described above. We characterize the $\lambda$-optimal sensor fusion algorithm in the form of a computable optimization problem.  The optimization requires knowledge of the underlying distribution and
has high computational complexity. In Section \ref{sec:practical}, we restrict our study to a specific subset of distributed estimation scenarios involving faulty sensor measurements, and provide practical fault-tolerant sensor fusion algorithms with low computational complexity whose performance is evaluated through simulations in Section \ref{sec:sim}.

\begin{Remark}
It can be noted that the objective function in the optimization in Equation \eqref{eq:ob} is a convex  function and the Hilbert space $\mathcal{L}_2(\prod_{i\in [n]}L_{i,j}\times U_{i,j})$ is a convex set. So, Equation \eqref{eq:ob} describes a convex optimization problem and a unique $\lambda$-optimal fusion algorithm $\textbf{f}_{\lambda}^*$ always exists.
\end{Remark}

\begin{Remark}
Define the set of all achievable accuracy and consensus values $\mathcal{O}\triangleq \{(m,c)\in \mathbb{R}^{2}| \exists \mathbf{f}\in \mathcal{L}_2(\prod_{i\in [n]}L_{i,j}\times U_{i,j}): mse(\mathbf{f})=m, cns(\mathbf{f})=c\}$, where $mse(\mathbf{f})\triangleq \sum_{j}mse(f_j)$ and $cns(\mathbf{f})\triangleq\sum_{j,j'}cns(f_j,f_{j'})$. It is straightforward to argue that $\mathcal{O}$ is a convex set, so that it is characterized by its supporting hyperplanes \cite{gray2017modern}. So, characterizing $\mathcal{O}$ is equivalent to solving $\eqref{eq:ob}$ for all $\lambda\in [0,1]$.
\end{Remark}

\begin{Remark}
\label{rem:1}
The $1$-optimal fusion algorithm is the optimal estimator in the absence of 
a consensus objective. It is well-known that, due to the orthogonality principle, we have $\textbf{f}^*_1= (\mathbb{E}(X| L_{i,j},U_{i,j}, i\in [n]), j\in [m])$, e.g., see \cite{durrett2019probability}. On the other hand, the $0$-optimal fusion algorithm only maximizes consensus, so that
all of the estimators output the same constant value. In the sequel, we wish to characterize the accuracy-consensus tradeoff for $\lambda\in (0,1)$.
\end{Remark}

\section{Analytical Derivation of the Optimal Fusion Algorithm}
\label{sec:analytical}
In this section, we derive an analytical expression of the $\lambda$-optimal fusion algorithm. The analytical solution requires knowledge of the underlying probability distribution defined in  \eqref{eq:dist}, and the optimization is computationally complex.  In Section \ref{sec:practical}, we address this by introducing low-complexity practical fusion algorithms with theoretical performance guarantees under specific statistical assumptions. The following provides the main result of this section.

\begin{Theorem}[\textbf{Optimal Fusion Algorithm}]
\label{th:1}
For a distributed estimation setup characterized by the tuple $(n,m, P_{X,({L}_{i,j},{U}_{i,j})_{i\in [n],j\in [m]}})$, and given $0\leq \lambda\leq 1$, the  $\lambda$-optimal fusions algorithm $\textbf{f}_{\lambda}^*=(f^*_{j}(\cdot), j\in [m])$ is given by $f^*_j(\cdot)=c^*_{j}\overline{f}_j^*(\cdot)+{b}^*_{j}, j\in [m]$ for constant vectors $\textbf{c}^*\triangleq (c^*_{1},c^*_{2},\cdots,c^*_{m})$, $\textbf{b}^*=(b^*_{1},b^*_2,\cdots,b^*_m)$ and zero-mean, unit-variance functions $\overline{\textbf{f}^*}(\cdot)\triangleq (\overline{f}^*_{1},\overline{f}^*_{2},\cdots,\overline{f}^*_{m}) $ given by:
\begin{align}
    \overline{\textbf{f}^*}= \argmax_{\overline{\textbf{f}}\in \overline{\mathcal{F}}}\Theta (AA^t)^{-1}\Theta^t,\quad \mathbf{c}^*= A^{-1}\Theta^t, \quad \textbf{b}^*=
    \mathbb{E}(X)\textbf{e}_m,
    \label{eq:th1}
\end{align}
where $A=[a_{j,j'}]_{j,j'\in [m]}$, $a_{j,j'}\triangleq\mathbbm{1}(j=j')-\frac{\overline{\lambda}}{(m-1)}\mathbbm{1}(j\neq j') \mathbb{E}(\overline{f}_j\overline{f}_{j'}), j,j'\in [m]$, $\Theta=[\theta_j]_{j\in [m]}$, $\theta_j\triangleq {\lambda}\mathbb{E}(X\overline{f}_j), j\in [m]$, and $\overline{\mathcal{F}}$ consists of  all vectors of zero-mean, unit-variance functions, i.e., $(\overline{f}_j(\cdot))_{j\in [m]}\in \overline{\mathcal{F}}$ iff $Var(\overline{f}_j(L_{i,j},U_{i,j}, i\in [n]))=1, \mathbb{E}(\overline{f}_j(L_{i,j},U_{i,j}, i\in [n]))=0 , j\in [m]$.
\end{Theorem}
 To prove the theorem, given a fusion algorithm $\textbf{f}(\cdot)=(f_j(\cdot),j\in [m])$, we decompose each $f_j(\cdot)$ into an amplitude $c_{j}$, a unit-variance function $\overline{f}_{j}$ capturing the \textit{direction} of $f_j(\cdot)$, and a bias $b_j$ as $f_j(\cdot)=c_{j}\overline{f}_j(\cdot)+b_j$, where $c_{j}>0$, $b_j\in \mathbb{R}$, $Var(\overline{f}_j(\cdot))=1$. First, we fix $\overline{\textbf{f}}(\cdot), j\in [m]$ and  $\textbf{c}$ and optimize $\mathbf{b}$ for a given $(\overline{\textbf{f}}(\cdot),\textbf{c})$. The component $c_{j}\overline{f}_j(\cdot)$ can be viewed as an element in the Hilbert space $\mathcal{L}_2(\prod_{i\in [n]}L_{i,j}\times U_{i,j})$, where $c_{j}$ is its amplitude and $\overline{f}_j(\cdot)$ determines its direction.

\begin{Proposition}[\textbf{Optimal Bias Vector}]
\label{prop:opt:bias}
For a fixed $\overline{\mathbf{f}}\in \overline{\mathcal{F}}$ and $\overline{\mathbf{c}}>0$, the bias vector $\textbf{b}^*$ optimizing Equation \eqref{eq:ob} is given by: 
\begin{align*}
\mathbf{b}^*=\mathbb{E}(X)\mathbf{e}_m.
\end{align*}
\end{Proposition}

\begin{proof}
The proof follows by noting that \[mse(f_j)= 
 \mathbb{E}\left(\left(X-\mathbb{E}(X)-\widehat{X}_j+\mathbb{E}(\widehat{X}_j)\right)^2\right)+ \left(\mathbb{E}(X)-\mathbb{E}(\widehat{X}_j)\right)^2,
 \]
which is minimized if $\mathbb{E}(X)=\mathbb{E}(\widehat{X}_j)$. Similarly, 
\[cns(f_j,f_{j'})= 
 \mathbb{E}\left(\left(\widehat{X}_j-\mathbb{E}(\widehat{X}_j)-\widehat{X}_{j'}+\mathbb{E}(\widehat{X}_{j'})\right)^2\right)+ \left(\mathbb{E}(\widehat{X}_j)-\mathbb{E}(\widehat{X}_{j'})\right)^2,
 \]
which is minimized if $\mathbb{E}(\widehat{X}_j)=\mathbb{E}(\widehat{X}_{j'})$. So, all terms in \eqref{eq:ob} are minimized simultaneously if $\mathbb{E}(\widehat{X}_j)=\mathbb{E}(X), j\in [m]$.
\end{proof}
In the rest of the paper, without loss of generality, we assume $\mathbb{E}(X)=0$ so that $\mathbf{b}^*=0$.
Next, we fix $\overline{f}_j(\cdot)$ and optimize the amplitude vector $\textbf{c}= (c_{j},j\in [m])$. The following proposition characterizes the optimal amplitude vector $\textbf{c}$. 
\begin{Proposition}[\textbf{Optimal Amplitude Vector}]
For a fixed $\overline{\mathbf{f}}\in \overline{\mathcal{F}}$, the amplitude vector $\textbf{c}^*$ optimizing  \eqref{eq:ob} is given by: 
\begin{align*}
    \textbf{c}^*=A^{-1}\Theta^t,
\end{align*}
where $A$ and $\Theta$ are defined in Theorem \ref{th:1}.
\end{Proposition}

\begin{proof}
Let us fix $\overline{\textbf{f}}$, and consider Equation \eqref{eq:ob}: 

\begin{align*}
   \argmin_{\textbf{c}}\lambda\!\!\sum_{j\in [m]} mse_j(c_{j}\overline{f}_j)+ \frac{\overline{\lambda}}{m-1}\!\!\sum_{\substack{j,j'\in [m]\\ j\neq j'}}  cns_{j,j'}(c_{j}\overline{f}_j,c_{{j'}}\overline{f}_{j'}).
\end{align*}
For $j\in [m]$, we take the derivative of the objective function with respect to $c_{j}$ and set it equal to 0. Note that 
\begin{align*}
\frac{\partial}{\partial c_{j}}mse(c_{j'}\overline{f}_{j'})
&= \mathbbm{1}(j=j')(2c_{j}\mathbb{E}(\overline{f}_j^2)-2\mathbb{E}(X\overline{f}_j)) = \mathbbm{1}(j=j')(2c_{j}-2\mathbb{E}(X\overline{f}_j))
.
\end{align*}
Furthermore, 
\begin{align*}
&\frac{\partial}{\partial c_{j}}cns(c_{j'}\overline{f}_{j'},c_{j''}\overline{f}_{j''})
= \mathbbm{1}(\exists j_d: \{j',j''\}=\{j,j_d\}) (2c_{j}-2c_{j_d}\mathbb{E}(\overline{f}_j\overline{f}_{j_d}))
.
\end{align*}
Let $L_{\lambda}$ denote the objective function in \eqref{eq:ob}. We have: 
\begin{align}
    \frac{\partial}{\partial c_{j}}L_{\lambda}= 2\lambda (c_{j}-\mathbb{E}(X\overline{f}_j))+\frac{2\overline{\lambda}}{m-1}\sum_{j'\neq j}(c_{j}-c_{j'}\mathbb{E}(\overline{f}_j\overline{f}_{j'})),
    \label{eq:aux1}
\end{align}
for all $ j\in [m]$.
This provides a system of $m$ linear equalities, and solving for $\mathbf{c}^*$ yields, $\mathbf{c}^*= A^{-1}\Theta^t$ as desired. 
\end{proof}
To complete the proof of Theorem \ref{th:1}, we find the optimal $\overline{\textbf{f}}(\cdot)$ given $\textbf{c}^*=A^{-1}\Theta^t$.

\begin{Proposition}[\textbf{Optimal Direction Vector}]
Given an arbitrary $\overline{\mathbf{f}}\in \overline{\mathcal{F}}$, let the corresponding amplitude vector be $\textbf{c}=A^{-1}\Theta^t$. Then, $\overline{\mathbf{f}}^*\in \overline{\mathcal{F}}$ optimizing \eqref{eq:ob} is given by: 
\begin{align*}
    \overline{\textbf{f}}^*= \argmax_{\overline{\textbf{f}}\in \overline{\mathcal{F}}}\Theta (AA^t)^{-1}\Theta^t,
\end{align*}
where $A$ and $\Theta$ are defined in Theorem \ref{th:1}.
\end{Proposition}
\begin{proof}
Let $L_{\lambda}$ denote the objective function in \eqref{eq:ob}. We have: 
\begin{align*}
   & L_{\lambda}= \lambda \sum_{j} \mathbb{E}(X^2)+\lambda\sum_{j} c^2_{j}-2\lambda \sum_{j} c_{j}\mathbb{E}(X\overline{f}_j)
  +2\overline{\lambda} \sum_{\substack{j}} c^2_{j}-2\frac{\overline{\lambda}}{m-1} \sum_{\substack{j,j'\\j\neq j'}} c_{j}c_{j'}\mathbb{E}(\overline{f}_{j}\overline{f}_{j'})
\end{align*}
Note that 
\begin{align*}
    &2\overline{\lambda} \sum_{\substack{j}} c^2_{j}-2\frac{\overline{\lambda}}{m-1} \sum_{\substack{j,j'\\j\neq j'}} c_{j}c_{j'}\mathbb{E}(\overline{f}_{j}\overline{f}_{j'})
    =
    \sum_{j} c_{j}( \sum_{j'\neq j} \frac{2\overline{\lambda}}{m-1}(2c_{j}- c_{j'}\mathbb{E}(\overline{f}_{j}\overline{f}_{j'}))
\end{align*}
Using \eqref{eq:aux1}, we can see that this is equal to $-2\lambda (c_{j}-\mathbb{E}(X\overline{f}_j))$. So, $
    L_{\lambda}= \lambda \sum_{j}(\mathbb{E}(X^2)- c^2_{j}).$
As a result, minimizing $L_{\lambda,}$ is equivalent to maximizing $\sum_{j}c^2_{j}$. Since $\textbf{c}=A^{-1}\Theta^t$, this is equivalent to maximizing $(A^{-1}\Theta^t)^tA^{-1}\Theta^t= \Theta (AA^t)^{-1}\Theta^t$. 
\end{proof}

From Theorem \ref{th:1} it follows that the optimal fusion algorithm is given by $\argmax_{\overline{\textbf{f}}\in \overline{\mathcal{F}}}\Theta (AA^t)^{-1}\Theta^t$. This optimization, taken over all zero-mean and unit-variance functions, can be reformulated as an optimization over real-valued matrices as follows. 
Let $\mathsf{F}=\{\overline{f}_{j,1},\overline{f}_{j,2},\cdots \}$ be an orthonormal basis for the Hilbert space $\mathcal{L}_2(\prod_{i\in [n]}L_{i,j}\times U_{i,j}), j\in [m]$. Then, the optimization can be re-written as follows:
\begin{align}
\overline{\textbf{f}^*}= \argmax_{(\epsilon_{j,k})_{j\in [m], k\in \mathbb{N}}: \sum_{k=1}^{\infty}\epsilon^2_{j,k}=1} \Theta(AA^t)^{-1}\Theta^t,
\label{eq:ob2}
\end{align}
where $\underline{f}_j(\cdot)= \sum_{k=1}^{\infty}\epsilon_{i,j}\underline{f}_{j,k}, j\in [m]$, so that $a_{j,j'}=\mathbbm{1}(j=j')-\mathbbm{1}(j\neq j')\frac{\overline{\lambda}}{m-1} \sum_{k,k'} \epsilon_{j,k}\epsilon_{j',k'}
\mathbb{E}(\underline{f}_{j,k}\underline{f}_{j',k'}), j,j'\in [m]$ and $\theta_j=\sum_{k} \epsilon_{j,k}\mathbb{E}(X\underline{f}_{j,k})$. 

As mentioned in Remark \ref{rem:1}, it is well-known that the $1$-optimal fusion algorithm is given as \[\textbf{f}^*_1= (\mathbb{E}(X| L_{i,j},U_{i,j}, i\in [n]), j\in [m]).\] This result can be re-derived from Equation \eqref{eq:ob2} by taking the orthonormal basis $\mathsf{F}$ such that \[\overline{f}_{j,1}= \frac{\mathbb{E}(X| L_{i,j},U_{i,j}, i\in [n])}{\sqrt{Var(\mathbb{E}(X| L_{i,j},U_{i,j}, i\in [n]))}}, j\in [m].\] It suffices to show that $\epsilon_{1,j}=1, \epsilon_{k,j}=0, j\in [m],k\neq 1$. To see this, note that since $\lambda=1$, we have $a_{j,j'}=\mathbbm{1}(j=j'), j,j'\in [m]$ so that $A$ is a diagonal matrix, and $\Theta=[\epsilon_{1,1}, \epsilon_{2,1},\cdots,\epsilon_{m,1}]$ due to the fact that $X$ is orthogonal to the subspace generated by $\{\overline{f}_{j,2},\overline{f}_{j,3},\cdots \}$ which follows by the smoothing property of expectation, the fact that $\overline{f}_{j,1}= \frac{\mathbb{E}(X| L_{i,j},U_{i,j}, i\in [n])}{\sqrt{Var(\mathbb{E}(X| L_{i,j},U_{i,j}, i\in [n]))}}$, and that $\mathsf{F}$ is an orthonormal basis. So, $\Theta(AA^t)^{-1}\Theta^t=\sum_{j\in [m]}\epsilon^2_{j,1}$, which is maximized by taking $\epsilon_{j,1}=1, \epsilon_{j,k}=0, j\in [m], k\neq j$, since we must have $\sum_{k=1}^{\infty} \epsilon_{j,k}^2=1, j\in [m]$ in the constrained optimization in Equation \eqref{eq:ob2}.

\subsection{Cramer-Rao Type Bound on Accuracy-Consensus Loss}
\label{sec:CR}
In the previous section, we characterized the $\lambda$-optimal fusion algorithm. Next,  we provide a Cramer-Rao type bound on the performance of the algorithm with respect to the accuracy-consensus loss. To present the resulting bound concisely, we make additional assumptions on the structure of the sensor  measurements' statistics.  
Specifically,  we assume that a subset $\tau\in [n]$ of the sensors are faulty, and the rest of the sensors are non-faulty. Faulty sensors report independent measurements to each agent, that is, if the $i$th sensor is faulty, the pair $(L_{i,j},U_{i,j})$ is independent of $(L_{i',j'},U_{i',j'}), (i,j)\neq (i',j')$ and $X$. Otherwise, if the sensor is non-faulty, the sensor reports the same measurement to all agents. That is, there exits $L_i,U_i$ such that $L_{i,j}=L_i, j\in [m]$ and $U_{i,j}=U_i, j\in [m]$, and $X$ is in the interval $[L_{i},U_{i}]$ with probability one. It is assumed that each sensor is equally likely to be faulty, i.e., the $\tau$ faulty sensors are chosen randomly and uniformly among the $n$ sensors. Formally, we consider a joint distribution $P_{X,(\widetilde{L}_{i},\widetilde{U}_{i})_{i\in [n]}}$ defined on $\mathbb{R}^{2n+1}$ such that $P(X\in [\widetilde{L}_i,\widetilde{U}_i])=1, i\in [n]$. The pair $(\widetilde{L}_{i},\widetilde{U}_{i}), i\in [n],j\in [m]$ represent the \textit{ground-truth} measurement at the $i$th sensor.
The joint measure on $X, (L_{i,j},U_{i,j})_{i\in [n],j\in [m]}$ is given as follows:
\begin{align}
    &\label{eq:dist}
    P_{X, (L_{i,j},U_{i,j})_{i\in [n],j\in [m]}}(\mathcal{A}, (\mathcal{L}_{i,j},\mathcal{U}_{i,j})_{i\in [n],j\in [m]})
    \\&\nonumber=\sum_{\underline{t}\in \mathcal{T}_{{\tau}}^n} \frac{1}{{n\choose \tau}}P_{X}(\mathcal{A})\prod_{j:t_j=1, i\in [m]} P_{\widetilde{L}_{i},\widetilde{U}_{i}}(\mathcal{L}_{i,j},\mathcal{U}_{i,j})\times 
    \\& \nonumber
    \prod_{j:t_j=0} P_{\widetilde{L}_{i},\widetilde{U}_{i}|X}(\bigcap_{i\in[n]}\mathcal{L}_{i,1},\bigcap_{i\in [n]}\mathcal{U}_{i,1}|\mathcal{A}),
\end{align}
where $\mathcal{A}, (\mathcal{L}_{i,j},\mathcal{U}_{i,j})\in \mathbb{B}$ and $\mathcal{T}_{\tau}^n$ is the set of n-length binary vectors with Hamming weight equal to $\tau$.

\begin{Theorem}
For a distributed estimation setup described above, assume that the non-faulty sensors' outputs are jointly independent  given $X$. For $\lambda\in [0,1]$, define the accuracy-consensus loss as:
\begin{align*}
   {L}^*_\lambda\triangleq\min_{\textbf{f}=(f_j)_{j\in [m]}}\lambda\!\!\sum_{j\in [m]} mse(f_j)+ \frac{\overline{\lambda}}{m-1}\!\!\sum_{\substack{j,j'\in [m]\\ j\neq j'}}  cns(f_j,f_{j'}).  
\end{align*}
  Then, 
\begin{align*}
    \frac{m^2\lambda}{\overline{I}_F(X)}\leq  {L}^*_\lambda,\quad  \lambda\in [0,1],
\end{align*}
 where $\overline{I}_F(X)$ is the average Fisher information defined as:
 \begin{align*}
   \overline{I}_F(X)\triangleq   \sum_{\underline{t}\in \mathcal{T}_{{\tau}}^n} \frac{1}{{n\choose \tau}} \sum_{j:t_j=0}  \mathbb{E}
   \left(\left(
   \frac{\partial \log{f_{\widetilde{L}_{t_j},\widetilde{U}_{t_j}|X}(\cdot|X)}
   }{\partial{X}}
   \right)^2\right), 
 \end{align*}
 where the expectation is taken over $X, (L_{i,j},U_{i,j})_{i\in [n],j\in [m]}$.
 Particularly, if the non-faulty sensors' outputs are identically distributed given X, we have:
  \begin{align*}
  \frac{m^2\lambda}{(n-\tau)\mathbb{E}_{X,\widetilde{L},\widetilde{U}}
   \left(\left(
   \frac{\partial \log{f_{\widetilde{L},\widetilde{U}|X}(\cdot|X)}
   }{\partial{X}}
   \right)^2\right)}\leq {L}^*_\lambda.
 \end{align*}
\end{Theorem}

\begin{proof}
The proof follows by steps similar to that of the Cramer-Rao bound \cite{van2004detection}. We provide an outline in the following. 
Let $f_j, j\in [m]$ be optimal fusion algorithms for a genie-aided version of the problem where the agents are aware of the faulty sensors. Consider the function:
\begin{align*}
    L(X, (L_{i,j},U_{i,j})_{i\in [n],j\in [m]})=\sqrt{\lambda}\!\!\sum_{j\in [m]} (f_j-X)+ \sqrt{\frac{\overline{\lambda}}{m-1}}\!\!\sum_{\substack{j,j'\in [m]\\ j\neq j'}}  (f_j-f_{j'}). 
\end{align*}
Note that by construction 
 for the optimal fusion algorithm we have $\mathbb{E}(L|X)=0$ since the optimal fusion algorithm is unbiased as shown in Proposition \ref{prop:opt:bias}. So, $\frac{\partial \mathbb{E}(L|X)}{\partial X}=0$. Note that $f_j, j\in [m]$ is not a function of $X$ and is only a  function of  $(L_{i,j},U_{i,j})_{i\in [n],j\in [m]}$. So, 
 \begin{align*}
    m \sqrt{\lambda}&= \int_{(l_{i,j},u_{i,j})_{i\in [n],j\in [m]}}L(X, (l_{i,j},u_{i,j})_{i\in [n],j\in [m]})
    \frac{\partial 
    f_{(L_{i,j},U_{i,j})_{i\in [n],j\in [m]}|X}((l_{i,j},u_{i,j})_{i\in [n],j\in [m]}|X)}{\partial X} d(l_{i,j},u_{i,j})_{i\in [n],j\in [m]} 
    \\&= 
    \int_{(l_{i,j},u_{i,j})_{i\in [n],j\in [m]}}L(X, (l_{i,j},u_{i,j})_{i\in [n],j\in [m]})
    \frac{\partial \prod_{i\in [n],j\in [m]}f_{L_{i,j},U_{i,j}|X}(l_{i,j},u_{i,j}|X)}{\partial X} d(l_{i,j},u_{i,j})_{i\in [n],j\in [m]} 
    \\&\leq \sum_{\underline{t}\in \mathcal{T}_{{\tau}}^n} \frac{1}{{n\choose \tau}} \sum_{j:t_j=0}
    \int_{(l_{t_j},u_{t_j})_{i\in [n],j\in [m]}}L(X, (l_{t_j},u_{t_j})_{j:t_j=0})
    \frac{\partial \prod_{j:t_j=0}f_{L_{t_j},U_{t_j}|X}(l_{t_j},u_{t_j}|X)}{\partial X} d(l_{t_j},u_{t_j})_{j: t_j=0} 
 \end{align*}
 where in the last step, we have used the genie-aided assumption to remove the inputs from faulty sensors. We bound each of the terms in the two summations in $L$
separately. For instance, let us consider $\sqrt{\lambda}(f_j-X)$ for some $j\in [m]$. Then, 
\begin{align*}
     &\sum_{\underline{t}\in \mathcal{T}_{{\tau}}^n} \frac{1}{{n\choose \tau}} \sum_{j:t_j=0}
    \int_{(l_{t_j},u_{t_j})_{i\in [n],j\in [m]}}\sqrt{\lambda}(f_j-X)
    \frac{\partial \prod_{j:t_j=0}f_{L_{t_j},U_{t_j}|X}(l_{t_j},u_{t_j}|X)}{\partial X} d(l_{t_j},u_{t_j})_{j: t_j=0} 
    \\& = 
    \sum_{\underline{t}\in \mathcal{T}_{{\tau}}^n} \frac{1}{{n\choose \tau}} \sum_{j:t_j=0}
    \int_{(l_{t_j},u_{t_j})_{i\in [n],j\in [m]}}\sum_{j:t_j\neq 0}\left(\sqrt{\lambda}(f_j-X)\sqrt{f_{L_{t_j},U_{t_j}|X}(l_{t_j},u_{t_j}|X)}\right)\times \\&\left(\frac{\partial \log{f_{{L}_{t_j},{U}_{t_j}|X}(l_{t_j},u_{t_j}|X)}
   }{\partial{X}}\sqrt{f_{L_{t_j},U_{t_j}|X}(l_{t_j},u_{t_j}|X)}\right)d(l_{t_j},u_{t_j})_{j: t_j=0} 
   \\& \leq{\lambda} \mathbb{E}((f_j-X)^2)
    \sum_{\underline{t}\in \mathcal{T}_{{\tau}}^n} \frac{1}{{n\choose \tau}} \sum_{j:t_j=0}
   \sum_{j:t_j\neq 0}\mathbb{E}\left(\left(\frac{\partial \log{f_{{L}_{t_j},{U}_{t_j}|X}(l_{t_j},u_{t_j}|X)}
   }{\partial{X}}\right)^2\right),
\end{align*}
where in the last step
 used the Cauchy-Shwarz inequality similar to the proof of the original Cramer-Rao bound. Combining the resulting terms in the summation completes the proof.
\end{proof}

\section{Practical Low-Complexity Fusion Algorithms}
\label{sec:practical}
In Section \ref{sec:analytical}, we evaluated the accuracy-consensus tradeoff under the general assumptions on the target and measurement statistics. The resulting optimization in Equation \eqref{eq:th1} requires knowledge of the underlying distribution and
has high computational complexity. In this section, we restrict our study to a specific subset of distributed estimation scenarios and provide practical fault-tolerant sensor fusion algorithms. 
Specifically, in the formulation considered in this section, it is assumed that a subset $\tau\in [n]$ of the sensors are faulty, and the rest of the sensors are non-faulty. Faulty sensors report independent measurements to each agent, that is, if the $i$th sensor is faulty, the pair $(L_{i,j},U_{i,j})$ is independent of $(L_{i',j'},U_{i',j'}), (i,j)\neq (i',j')$ and $X$. Otherwise, if the sensor is non-faulty, the sensor reports the same measurement to all agents. That is, there exits $L_i,U_i$ such that $L_{i,j}=L_i, j\in [m]$ and $U_{i,j}=U_i, j\in [m]$, and $X$ is in the interval $[L_{i},U_{i}]$ with probability one. It is assumed that each sensor is equally likely to be faulty, i.e., the $\tau$ faulty sensors are chosen randomly and uniformly among the $n$ sensors. To elaborate, we consider a joint distribution $P_{X,(\widetilde{L}_{i},\widetilde{U}_{i})_{i\in [n]}}$ defined on $\mathbb{R}^{2n+1}$ such that $P(X\in [\widetilde{L}_i,\widetilde{U}_i])=1, i\in [n]$. The pair $(\widetilde{L}_{i},\widetilde{U}_{i}), i\in [n],j\in [m]$ represent the \textit{ground-truth} measurement at the $i$th sensor.
The joint measure on $X, (L_{i,j},U_{i,j})_{i\in [n],j\in [m]}$ is given as follows:
\begin{align}
    &\label{eq:dist}
    P_{X, (L_{i,j},U_{i,j})_{i\in [n],j\in [m]}}(\mathcal{A}, (\mathcal{L}_{i,j},\mathcal{U}_{i,j})_{i\in [n],j\in [m]})
    \\&\nonumber=\sum_{\underline{t}\in \mathcal{T}_{{\tau}}^n} \frac{1}{{n\choose \tau}}P_{X}(\mathcal{A})\prod_{j:t_j=1, i\in [m]} P_{\widetilde{L}_{i},\widetilde{U}_{i}}(\mathcal{L}_{i,j},\mathcal{U}_{i,j})
    \prod_{j:t_j=0} P_{\widetilde{L}_{i},\widetilde{U}_{i}|X}(\bigcap_{i\in[n]}\mathcal{L}_{i,1},\bigcap_{i\in [n]}\mathcal{U}_{i,1}|\mathcal{A}),
\end{align}
where $\mathcal{A}, (\mathcal{L}_{i,j},\mathcal{U}_{i,j})\in \mathbb{B}$ and $\mathcal{T}_{\tau}^n$ is the set of n-length binary vectors with Hamming weight equal to $\tau$.

 A common approach in the estimation literature is to restrict the search to specific classes of estimation algorithms, e.g. linear and limited-degree polynomial estimation algorithms. To this end, Equation \eqref{eq:ob2} can be used to find the optimal fusion algorithm over a given subspace of the  $\mathcal{L}_2(\prod_{i\in [n]}L_{i,j}\times U_{i,j}), j\in [m]$ by choosing the orthonormal basis, $\mathsf{F}$, as the basis for the desired subspace. In the sequel, we consider several classes of previously studied fusion algorithms, introduce a new class of algorithms called the generalized Brooks-Iyengar Algorithms (GBI), and evaluate their performance in a specific scenario to provide insights into this optimization process. 
\begin{Definition}[\textbf{Classes of Fusion Algorithms}]
For the distributed estimation setup described above, we define the following class of fusion algorithms:
\begin{itemize}
    \item \textbf{Linear Fusion Algorithms:} The class of linear fusion algorithms $\mathsf{F}_{lin}$ is defined as: 
    \begin{align}
      &(f_{j}(\cdot))_{j\in[m]}
      \in  \mathsf{F}_{lin} \iff  \exists \epsilon_{i,j},\delta_{i,j},\gamma_j\in \mathbb{R}:
     \label{eq:lin} f_j((L_{i,j},U_{i,j})_{i\in [n]})=\! \sum_{i=1}^n\! \epsilon_{i,j}L_{i,j}\!+\!\delta_{i,j}U_{i,j}+\gamma_j, \hspace{0.05in} j\in [m].   
    \end{align}
% \item \textbf{Quadratic Fusion Algorithms:}
% \begin{align*}
%       &(f_{j}(\cdot))_{j\in[m]}
%       \in  \mathsf{F}_{quad} \iff  \exists \epsilon_{1,i,j},\epsilon_{2,i,j}, \delta_{1,i,j}, \delta_{2,i,j}\in \mathbb{R}:
%     \\& f_j((L_{i,j},U_{i,j})_{i\in [n]})= \sum_{i=1}^n\sum_{k=1}^2 \epsilon_{k,i,j}L^k_{i,j}+\delta_{k,i,j}U^k_{i,j}, j\in [m].   
%     \end{align*}
\item \textbf{Marzullo's Algorithm\cite{marzullo1990tolerating}:} Let $L_{i_1,j}\leq L_{i_2,j}\leq \cdots\leq L_{i_m,j}$ and $U_{i_1,j}\leq U_{i_2,j}\leq \cdots\leq U_{i_m,j}$. Then, 
\begin{align*}
    f_j((L_{i,j},U_{i,j})_{i\in [n]})= \frac{L_{\tau+1,j}+U_{n-\tau-1,j}}{2},\quad j\in [m]
\end{align*}

\item \textbf{Brooks-Iyengar (BI) Algorithm \cite{brooks1996robust}:} Let $g_j(x)\triangleq \sum_{i=1}^n \mathbbm{1}(x\in [L_{i,j},U_{i,j}]), j\in [m]$ and let the  $x_{j,1}<x_{j,2}<\cdots<x_{j,K}$ be all points of transition of $g_j(x)$, i.e., $\forall k: \lim_{x\to x_{j,k}^{-}}g_j(x) \neq \lim_{x\to x_{j,k}^{+}}g_j(x)$. Then, 
\begin{align*}
    f_j((L_{i,j},U_{i,j})_{i\in [n]})= \frac{\sum_{k\in [K-1]: g_{j}(\overline{x}_{j,k})\geq n-\tau} \overline{x}_{j,k}g_j(\overline{x}_{j,k})} {\sum_{k\in [K-1]: g_{j}(\overline{x}_{j,k})\geq n-\tau} g_j(\overline{x}_{j,k})},\quad j\in [m],
\end{align*}
where $ \overline{x}_{j,k}= \frac{x_{j,k}+x_{j,k+1}}{2}, k\in [K-1]$.
\item \textbf{Generalized Brooks-Iyengar Algorithms:} The class of GBI algorithms $\mathsf{F}_{GBI}$ is defined as:
\begin{align}
      \nonumber&(f_{j}(\cdot))_{j\in[m]}
      \in  \mathsf{F}_{GBI} \iff  \exists w_{\underline{t}}\in \mathbb{R}:
   f_j((L_{i,j},U_{i,j})_{i\in [n]})= 
    \frac{  \sum_{\underline{t}\in \mathcal{T}^n_\tau}
    w_{\underline{t}}m_{\underline{t}}}{\sum_{\underline{t}\in \mathcal{B}(n,\tau)} w_{\underline{t}}},\quad j\in [m],   
    \end{align}
    where $m_{\underline{t}}=\frac{|\min_{j:t_j=0} u_j + \max_{j:t_j=0} \ell_j|}{2}$ and $\mathcal{T}_{\tau}^n=\{\underline{t}\in \{0,1\}^n| \sum_{i=1}^n t_i=\tau\}$.
    
\end{itemize}
\end{Definition}

\begin{Remark}
\label{rem:sym}
In the distributed estimation setup described by Equation \eqref{eq:dist} any $\lambda$-optimal linear fusion algorithm must have $\epsilon_{i,j}=\epsilon_{i,j'}$ and $\delta_{i,j}=\delta_{i,j'}, i\in [n], j,j'\in [m]$. The reason is that the objective function in \eqref{eq:ob2} is convex and it is symmetric with respect to the  variables $\epsilon_{i,j}, \epsilon_{i,j'}$ and $\delta_{i,j},\delta_{i,j'}$ due to the probability distribution give in \eqref{eq:dist}. 
%Similarly, for quadratic fusion algorithms, we have $\epsilon_{k,i,j}=\epsilon_{k,i,j'}$ and $\delta_{k,i,j}=\delta_{k,i,j'}, k\in \{1,2\}, i\in [n], j,j'\in [m]$.
\end{Remark}
\begin{Remark}
 Finding the optimal GBI fusion algorithm requires solving the following:
\begin{align*}
\overline{\textbf{f}^*}_{GBI}= \argmax_{w_{\underline{t}}: \sum_{\underline{t}}w^2_{\underline{t}}=1} \Theta(AA^t)^{-1}\Theta^t,
\end{align*}
where $\underline{f}_j(\cdot)= \sum_{\underline{t}}w_{\underline{t}}\frac{m_{\underline{t}}}{Var(m_{\underline{t}})}, j\in [m]$.
\end{Remark}

The following proposition characterizes the $\lambda$-optimal linear fusion algorithm for $m=2$ and $\mathbb{E}(X)=0$. The characterization is used in the numerical simulations in the subsequent section.

\begin{Proposition}
\label{prop:opt_lin}
Consider the distributed estimation setup described by Equation \eqref{eq:dist}. Let $m=2$, $\mathbb{E}(X)=0$, and $\lambda\in (0,1)$. The parameters $(\epsilon_{i,j},\delta_{i,j},\gamma_j), i\in [n], j\in [m]$ characterizing the $\lambda$-optimal linear fusion satisfy:
\begin{enumerate}
    \item $\exists (\epsilon_j,\delta_j): \epsilon_j\!=\epsilon_{i,j},\quad \delta_j=\delta_{i,j},\quad i\in[n],j\in \{1,2\}$, 
    \item $\gamma_j=-n(\epsilon_j\mathbb{E}(L_1)\!+\!\delta_j\mathbb{E}(U_1)),\quad  j\in \{1,2\}$,
    \item $\delta_j$ is a root of  $\xi_{1,j}x^2+\xi_{2,j}x+\xi_{3,j}$, where
    \begin{align*}
        &\xi_{1,j}\triangleq nVar(U_1)+n(n-1)Cov(U_1,U_2)\\
        &\xi_{2,j}\triangleq 2\epsilon_j( nCov(L_1,U_1)+n(n-1)Cov(L_1,U_2))\\
        &\xi_{3,j}\triangleq nVar(L_1)+n(n-1)Cov(L_1,L_2)
    \end{align*}
    \item $(\epsilon_1,\epsilon_2)$ is given as:
    \begin{align*}
        (\epsilon_1,\epsilon_2)=\argmin_{\epsilon_1,\epsilon_2\in \mathbb{R}^2} \frac{\theta_1^2+\theta_2^2}{1-z^2}+ \frac{2z(\theta_1+\theta_2)^2}{(1-z^2)^2}, 
    \end{align*}
    where $\theta_j\triangleq n\epsilon_jCov(L_1,X)+n\delta_jCov(U_1,X)$ and
    \begin{align*}
        z\triangleq &-(1-\lambda)\times  \Big(\epsilon_1\epsilon_2(Var(L_1)+n(n-1)Cov(L_1,L_2))+
        \delta_1\delta_2(Var(U_1)+n(n-1)Cov(U_1,U_2))+
        \\&(\epsilon_1\delta_2+\epsilon_2\delta_1)(nCov(L_1,U_1)+n(n-1)Cov(L_1,U_2))\Big).
    \end{align*}
\end{enumerate}
\end{Proposition}
The proof follow by optimizing Equation \eqref{eq:ob2} on the set of linear fusion algorithms. We provide a summary of the proof arguments in the following: 1) follows by Remark \ref{rem:sym},  2) follows from the fact that $\mathbb{E}(\overline{f}_j)= 0$ and $\mathbf{b}= \mathbb{E}(X)\mathbf{e}_m=0$, 3) follows from the fact that $\mathbb{E}(\overline{f}^2_j)= 1$, and 4) follows by setting $A= \begin{bmatrix}
1 & -z\\
-z & 1
\end{bmatrix}$ and simplifying the optimization in Equation \eqref{eq:ob2}.

\section{Numerical Example: Uniformly Distributed Variables}
\label{sec:sim}
In order to  numerically study the  accuracy-consensus tradeoff characterized in the prequel, in this section we consider a specific distributed estimation example and perform numerical simulations of various fusion algorithms. 
In particular, we consider the distributed estimation setup described by Equation \eqref{eq:dist} and assume that $X$ is uniformly distributed over the interval $[-x_{max}, x_{max}]$, where $x_{max}\in \mathbb{N}$. Let the \textit{precision} of the $j$th sensor be parametrized by the random variable $\delta_j$ which is uniformly distributed over the set $\{1,2,\cdots,x_{max}\}$. 
Let the set of possible outputs for the $j$th sensor be $\mathsf{P}_j= \big\{[-x_{max}+2\frac{d-1}{\delta_j}x_{max} , -x_{max}+2\frac{d}{\delta_j}x_{max}], d\in [\delta_j]\big\}$. Note that by this construction, $\widetilde{L}_{i},\widetilde{U}_i, i\in [n]$ are discrete variables taking value from $\mathcal{L}=\big\{-x_{max}+2\frac{d-1}{\delta_j}x, d\in [\delta_j]\big\}$ and $\mathcal{U}=\big\{-x_{max}+2\frac{d}{\delta_j}x_{max}, d\in [\delta_j]\big\}$, respectively. Given $X\in [-x_{max},x_{max}]$, the variables $\widetilde{L}_{i},\widetilde{U}_i, i\in [n]$ take the unique values in $\mathcal{L}$ and $\mathcal{U}$, respectively, for which $X\in [\widetilde{L}_{i},\widetilde{U}_i]$. 

% Under the above statistical model, i)  X is uniformly distributed, i.e., $f_X\sim Uniform[-x,x], x>0$,  ii) Given $\widetilde{L}_{i}=\ell$ and $\widetilde{U}_{i}=u$, X is uniformly distributed on the interval $[\ell,u]$, i.e. $f_{X|\widetilde{L}_{i},\widetilde{U}_{i}}(\cdot|\ell,u)\sim Uniform[\ell,u], \ell<u, i\in [n]$, and iii) sensor measurements are independent of each other given X, i.e., the Markov chain $\widetilde{L}_{i},\widetilde{U}_i\leftrightarrow X \leftrightarrow (\widetilde{L}_{i'},\widetilde{U}_i', i'\neq i)$ holds for $i\in [n]$. %To this end, fix $M\in \mathbb{N}$ and let the precision of the $j$th sensor be parametrized by the random variable $\delta_j$ which is uniformly distributed over a set $\{1,2,\cdots,M\}$. Define the set of possible outputs for the $j$th sensor as $\mathsf{P}_j= \{[-x+\frac{d-1}{\delta_j}2x , -x+\frac{d}{\delta_j}2x], d\in [\delta_j] \}$. Note that in this case $\widetilde{L}_{i},\widetilde{U}_i, i\in [n]$ are discrete variables taking value from $\{-x+\frac{d-1}{\delta_j}2x, d\in [\delta_j]\}$ and $\{-x+\frac{d}{\delta_j}2x, d\in [\delta_j]\}$, respectively. 
% Define the joint distribution: 
% \begin{align*}
% P_{X,(\widetilde{L}_{i},\widetilde{U}_i)_{i\in [n]})}(\mathcal{A}, (\ell_i, u_i)_{i\in [n]})=\int_{\mathcal{A}\cap [-x,x]} \frac{1}{2x}\prod_{i\in [n]} \mathbbm{1}(x\in [\ell_i,u_i])dx'.
% \end{align*}
% It is straightforward to verify that this join distribution satisfies properties i)-iii) above. 
\begin{figure}%
    \centering
    \subfloat{{\includegraphics[width=7.5cm]{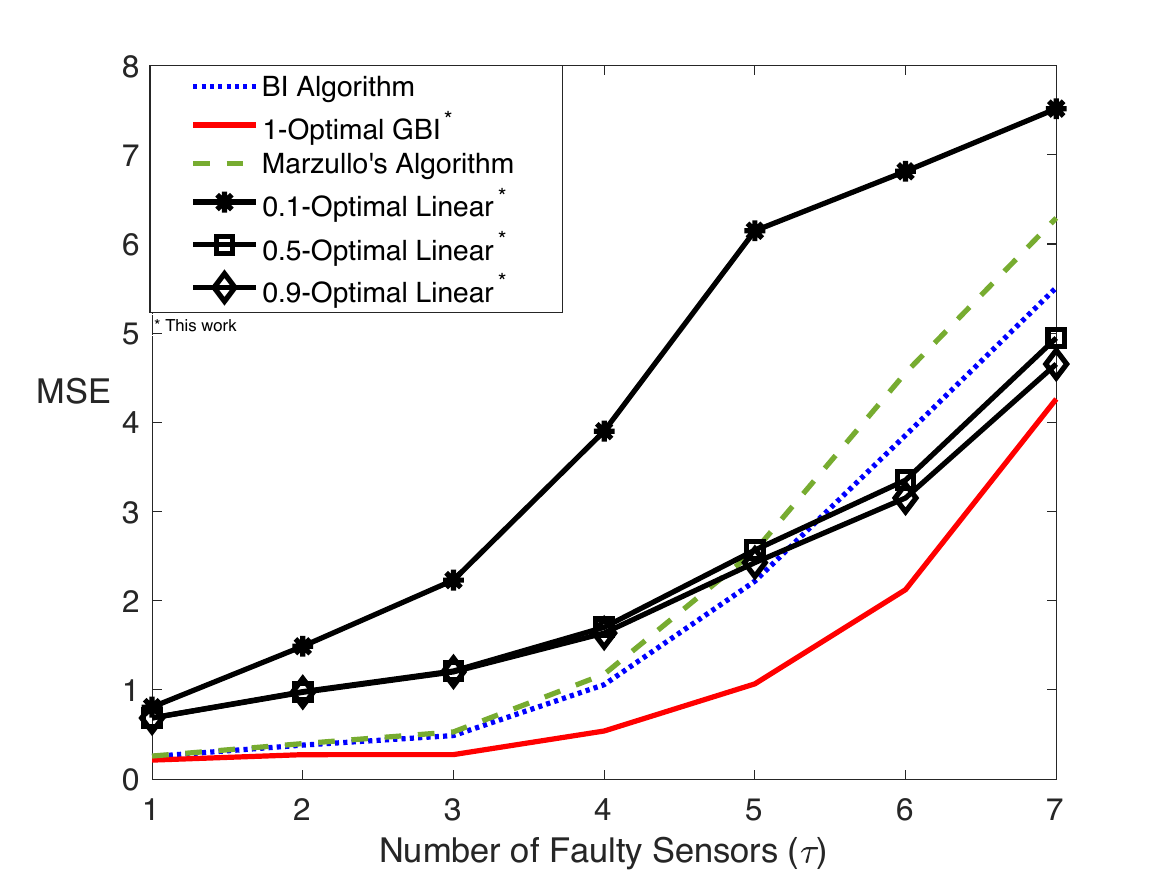} }}%
    \subfloat{{\includegraphics[width=7.5cm]{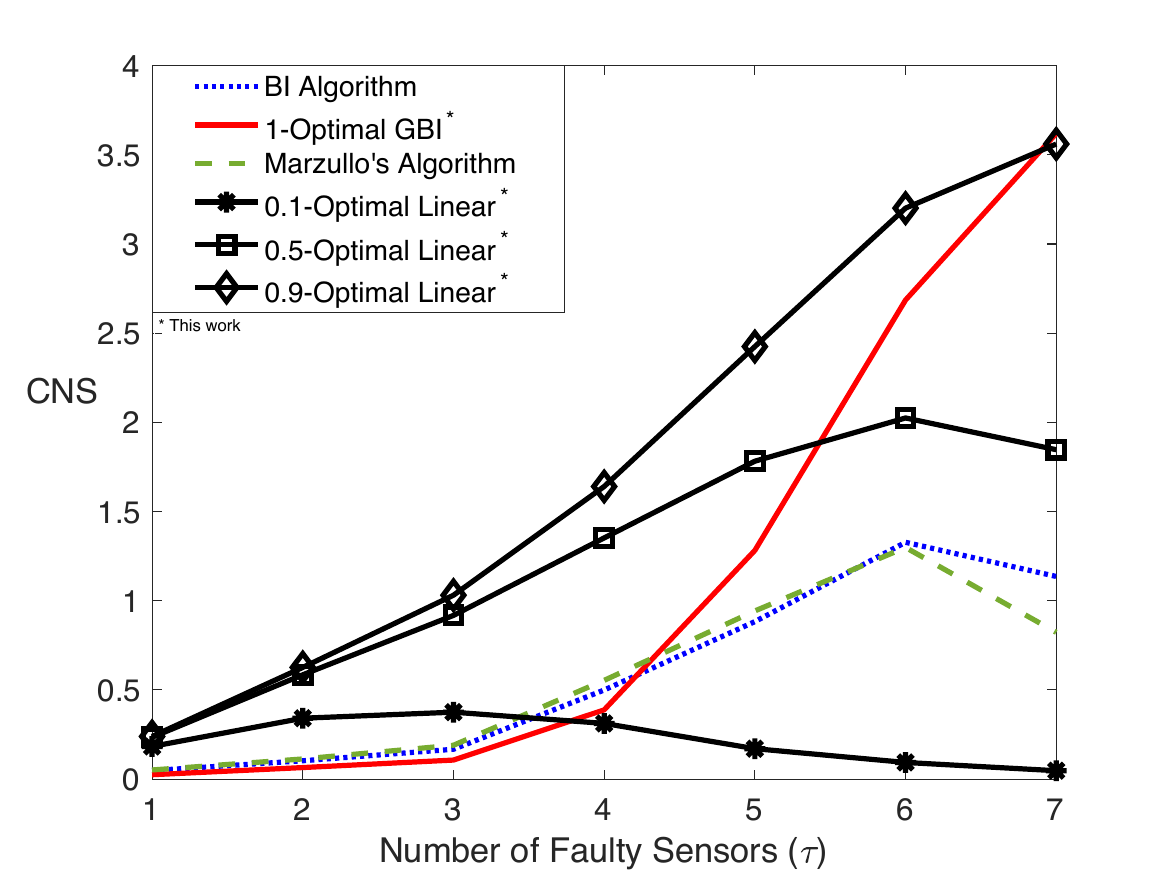} }}%
    \caption{Comparison of sensor fusion algorithms for $n=10$ sensors, $m=2$ agents, $\tau\in [7]$ faulty sensors, and $X\in [-5,5]$. }%
   \label{fig:compare}
\end{figure}

The following proposition shows that the $1$-optimal fusion algorithm for this setup is a GBI algorithm.

\begin{Proposition}
\label{prop:opt}
For the distributed estimation setup described above, the $1$-optimal fusion algorithm is equal to the GBI algorithm with  $w_{\underline{t}}\triangleq |\min u_j - \max \ell_j|^+\prod_{j:i_j=0} \frac{1}{u_j-\ell_j}, \underline{t}\in \mathcal{T}^n_{\tau}$. 
\end{Proposition}

\begin{proof}
The $1$-optimal fusion algorithm is given by   $\mathbb{E}(X|L_{i,j}^n,U_{i,j}^n)$. We focus on $j=1$ and drop the subscript $j$ when there is no ambiguity. We have:
\begin{align*}
  &  \mathbb{E}(X|L^n=\ell^n,U^n=u^n)= \int_x x f_{X|L^n,U^n}(x|\ell^n,u^n)dx
    \\&=  \int_x x \prod_{j=1}^n f_{L_j,U_j}(\ell_j,u_j) \frac{1}{{n \choose \tau}}
\sum_{\underline{t}\in \mathcal{B}(n,\tau)}
 \frac{\prod_{j:i_j=0} f_{X|L_j,U_j,i_j}(x|\ell_j,u_j,0)}{f^{n-\tau-1}_X(x)}dx
 \\&= \!\int_x x  \frac{\prod_{j=1}^n f_{L_j,U_j}(\ell_j,u_j)}{f_{L^n,U^n}(\ell^n,u^n)} \frac{1}{{n \choose \tau}}\!
\sum_{\underline{t}\in \mathcal{B}(n,\tau)}\!\!
 \frac{\prod_{j:i_j=0} f_{X|L_j,U_j,i_j}(x|\ell_j,u_j,0)}{f^{n-\tau-1}_X(x)}dx
 \\&= C\sum_{\underline{t}\in \mathcal{B}(n,\tau)}  \int_x x \prod_{j:i_j=0} \frac{\mathbbm{1}(x\in [\ell_j,u_j])}{u_j-\ell_j}dx,
\end{align*}
where $C\triangleq  \ell^{n-\tau-1}\frac{\prod_{j=1}^n f_{L_j,U_j}(\ell_j,u_j)}{f_{L^n,U^n}(\ell^n,u^n)} \frac{1}{{n \choose \tau}}$. Note that:
\begin{align*}
\nonumber
&    f_{X,L^n,U^n}(x,\ell^n,u^n)= \sum_{\underline{t}\in \mathcal{B}(n,\tau)} \frac{1}{{n \choose \tau}} f_{X,L^n,U^n|\underline{\mathcal{I}}}(x,\ell^n,u^n| \underline{t})
= \sum_{\underline{t}\in \mathcal{B}(n,\tau)} \frac{1}{{n \choose \tau}} f_{X}(x)f_{L^n,U^n|X,\underline{\mathcal{I}}}(\ell^n,u^n| x,\underline{t})
\\&\nonumber \stackrel{(a)}{=}\frac{1}{{n \choose \tau}}f_{X}(x)
\sum_{\underline{t}\in \mathcal{B}(n,\tau)}
\prod_{j:i_j=1} P_{L_j,U_j}(\ell_j,u_j)
\prod_{j:i_j=0} P_{L_j,U_j|X,i_j}(\ell_j,u_j|x,0)
\end{align*}
\begin{align*}
&\nonumber\stackrel{(b)}{=} \frac{1}{{n \choose \tau}}f_{X}(x)
\sum_{\underline{t}\in \mathcal{B}(n,\tau)}
\prod_{j:i_j=1} f_{L_j,U_j}(\ell_j,u_j)\times
\\&
\prod_{j:i_j=0} \frac{f_{X|L_j,U_j,i_j}(x|\ell_j,u_j,0)f_{L_j,U_j}(\ell_i,u_i)}{f_X(x)}
= \prod_{j=1}^n f_{L_j,U_j}(\ell_j,u_j) \frac{1}{{n \choose \tau}}
\sum_{\underline{t}\in \mathcal{B}(n,\tau)}
 \frac{\prod_{j:i_j=0} f_{X|L_j,U_j,i_j}(x|\ell_j,u_j,0)}{f^{n-\tau-1}_X(x)},
\end{align*}
where in (a) we have used the fact that the output of faulty sensors is independent of $X$, and (b) follows from the Bayes rule.
We have:
\begin{align*} 
&\frac{f_{L^n,U^n}(\ell^n,u^n)}{\prod_{j=1}^n f_{L_j,U_j}(\ell_j,u_j) }= \frac{\ell^{n-\tau-1}}{{n \choose \tau}}
\sum_{\underline{t}\in \mathcal{B}(n,\tau)}
\int_{x}\prod_{j:i_j=0} \frac{\mathbbm{1}(x\in [\ell_j,u_j])}{u_j-\ell_j}
= \frac{\ell^{n-\tau-1}}{{n \choose \tau}}
\sum_{\underline{t}\in \mathcal{B}(n,\tau)}
\frac{\min_{j:i_j=0} u_j - \max_{j:i_j=0} \ell_j}{\prod_{j:i_j=0} (u_j-\ell_j)}
\end{align*}
So, 
\begin{align}
C=\left(\sum_{\underline{t}\in \mathcal{B}(n,\tau)}
\frac{|\min_{j:i_j=0} u_j - \max_{j:i_j=0} \ell_j|^+}{\prod_{j:i_j=0} (u_j-\ell_j)}\right)^{-1}.
\label{eq:3}
\end{align}
Furthermore,:
\begin{align}
   & \mathbb{E}(X|L^n=\ell^n,U^n=u^n)= C\sum_{\underline{t}\in \mathcal{B}(n,\tau)}  \int_{x\in \cap_{j: i_j=0} [\ell_j,u_j]}\!\! x \prod_{j:i_j=0} \frac{1}{u_j-\ell_j}dx.
\nonumber
 \\&=   C\sum_{\underline{t}\in \mathcal{B}(n,\tau)}  
    \frac{|\min u_j - \max \ell_j|^+(\min u_j + \max \ell_j)}{2}\prod_{j:i_j=0} \frac{1}{u_j-\ell_j}.
    \label{eq:2}
\end{align}
Combining \eqref{eq:3} and \eqref{eq:2}, we have:
 \begin{align}
 \nonumber
      \mathbb{E}(X|L^n=\ell^n,U^n=u^n)&= 
    \frac{  \sum_{\underline{t}\in \mathcal{B}(n,\tau)}
    w_{\underline{t}}m_{\underline{t}}}{\sum_{\underline{t}\in \mathcal{B}(n,\tau)} w_{\underline{t}}},
 \end{align}
\end{proof}

In order to compare the performance of the sensor fusion algorithms introduced in Section \ref{sec:practical}, we numerically simulate their performance in a scenario with $n=10$ sensors, $m=2$ agents, $1\leq \tau\leq 7$ faulty sensors, and $x_{max}=5$. We consider the $\lambda$-optimal linear fusion algorithm given in Proposition \ref{prop:opt_lin} for $\lambda\in\{0.1,0.5,.0.9\}$, the original BI algorithm \cite{brooks1996robust}, Marzullo's algorithm \cite{marzullo1990tolerating}, and the 1-optimal GBI algorithm introduced in Proposition \ref{prop:opt}.
Figure \ref{fig:compare} shows the resulting mean square error (MSE) and consensus costs (CNS). It can be noted that the $0.1$-optimal linear fusion algorithm has the worst MSE and the best CNS among the simulated algorithms. This is expected, as the fusion algorithm  prioritizes consensus over accuracy. On the other hand, the $0.9$-optimal linear fusion algorithm has the best MSE and worst CNS among the three linear fusion algorithms since it prioritizes accuracy over consensus.
Furthermore, the BI and Marzullo's algorithms perform well for $\tau\leq \frac{n}{3}$. This is in agreement with prior works (e.g. \cite{ao2016precision}) which provide worst-case performance guarantees for BI and Marzullo's algorithms when $\tau\leq \frac{n}{3}$. The simulation shows that in this scenario, the average-case performance is good  as well. The 1-optimal GBI has the best MSE performance as it is the optimal sensor fusion algorithm in terms of MSE as shown in Proposition \ref{prop:opt}. It also outperforms the BI and Marzullo's algorithms in terms of CNS for $\tau\leq \frac{n}{3}$. It is of note that the highly non-linear BI, Maruzllo, and GBI algorithms outperform the best linear estimators in terms of MSE in this scenario.

\section{Conclusion}
Distributed estimation in the context of sensor networks was considered, where a subset of sensor measurements are faulty. Faulty sensors model both sensor malfunctions, as well as adversarial interference in measurement and transmission phases.
 It was shown that there is an inherent tradeoff between satisfying the accuracy and consensus objectives. A computable characterization of the fusion algorithm optimizing this tradeoff was provided. Several classes of fusion algorithms were studied, and the theoretical derivations were verified through various numerical simulations of their performance in terms of accuracy and consensus objectives.

\noindent \textbf{Acknowledgements:}
The authors would like to thank Prof. Azad Madni and Prof.  S. Sandeep Pradhan for stimulating discussions.

\bibliographystyle{unsrt}

\end{document}